\documentclass[prl,twocolumn]{revtex4}
\usepackage{amsfonts,amssymb,amsmath}            
\usepackage[]{graphics,graphicx,epsfig}            
\usepackage{amsthm,bbold}
\newtheorem{theorem}{Theorem}

\newtheorem{lemma}{Lemma}

\newtheorem{protocol}{Protocol}
\newcommand{\ket}[1]{\left | #1 \right\rangle}
\newcommand{\bra}[1]{\left \langle #1 \right |}
\newcommand{\half}{\mbox{$\textstyle \frac{1}{2}$}}
\newcommand{\smallfrac}[2][1]{\mbox{$\textstyle \frac{#1}{#2}$}}
\newcommand{\Tr}{\text{Tr}}

\newcommand{\identity}{\mathbb{1}}
\renewcommand{\epsilon}{\varepsilon}
\begin{document}

\title{Resources for Entanglement Distribution via the Transmission of Separable States}
\date{\today}

\author{Alastair \surname{Kay}}
\affiliation{Centre for Quantum Technologies, National University of Singapore, 
			3 Science Drive 2, Singapore 117543}
\affiliation{Keble College, Parks Road, Oxford, OX1 3PG, UK}
\begin{abstract}
One of the many bizarre features of entanglement is that Alice, by sending a qubit to Bob in a separable state, can generate some entanglement between herself and Bob. This protocol is stripped down to the bare essentials to better elucidate the key properties of the initial resource state that enable this entanglement distribution. The necessary and sufficient conditions under which the correlations of a Bell-diagonal state serve as a useful resource are proven, giving upper and lower bounds on the entanglement that can be distributed when those conditions are met.
\end{abstract}

\maketitle

What is it about a given quantum state that activates any one of a number of strange quantum features? This study of the resources required to achieve different information processing tasks is at the heart of Quantum Information. Such tasks are naturally specified by the set of operations $\Lambda$ which can legally be implemented, and the corresponding states that provide a resource for achieving said task. One is particularly interested in the states $\Sigma_\Lambda$ which cannot achieve the desired result, and quantifying how useful other states are. This is naturally measured by the distance from the set $\Sigma_\Lambda$, 
$$
I_\Lambda=\min_{\sigma\in\Sigma_{\Lambda}}S(\rho\|\sigma)
$$
which uses the relative entropy $S(\rho\|\sigma)=\Tr(\rho\log_2\rho-\rho\log_2\sigma)$. Common examples include the restriction to Local Operations and Classical Communication, in which $\Sigma_{\text{LOCC}}$ are just the separable states, and the useful resource is the entanglement of the state. Similarly, reference frames provide a resource for overcoming the restrictions imposed by super-selection rules \cite{frames}.

One of the more intriguing, counter-intuitive, protocols to arise in recent years hints at a new classification of resources. Two parties who share a separable state can distribute entanglement between them by transmitting another separable state \cite{toby}! This protocol starts from a state $\rho_{ABC}$ which is initially partitioned between two parties, Alice and Bob, as $\rho_{AC|B}$. Alice then sends qubit $C$ to Bob. During transmission, it is required that $C$ is separable from everything else i.e.\ the bipartitioning $\rho_{C|AB}$ is separable. By the end, Alice and Bob hold $\rho_{A|BC}$, which we wish to be entangled. The correlations of $\rho_{ABC}$ constitute a resource for entanglement distribution by separable states (EDSS), and the protocol potentially provides practical benefits -- the correlations that constitute the resource for EDSS may be less susceptible to noise than the extremely fragile entanglement.

At this level, EDSS is a direct consequence of the existence of multipartite bound entanglement \cite{dur} i.e.\ a state $\rho_{ABC}$ which is separable under the bipartitions $C|AB$ and $B|AC$ may be entangled under the partition $A|BC$. In a one parameter system (e.g. temperature of a thermal state \cite{Kay:2010,Kay:2010b}), the existence of multipartite bound entanglement is not surprising. Indeed, it would be quite remarkable if, for every state, every possible bipartition were to become separable at the same parameter value.

Unlike recent work \cite{bound1,bound2}, studying the general question of bounding the entanglement change arising from a state $\rho_{ABC}$ which may initially be entangled, in the present paper, we focus more specifically on what it is in the correlations between Alice and Bob that permit EDSS. To this end, we reduce the protocol to its bare essentials.

\begin{protocol}EDSS:
\begin{enumerate}
\item Alice and Bob start with a separable state of two qubits, $\rho_{AB}$.
\item Alice introduces an ancilla, $C$, which is completely uncorrelated from $\rho_{AB}$. Without loss of generality, we take this to be $\rho_C=\half(\identity+s X)$ ($X$ is the Pauli-X matrix).
\item Alice performs a unitary $U_{AC}$, producing $\rho_{ABC}=U_{AC}\rho_{AB}\otimes\rho_C U_{AC}^\dagger$, but has selected $s$ to ensure that the bipartition $C|AB$ remains separable.
\item Alice sends the separable qubit $C$ to Bob.
\end{enumerate}
\end{protocol}
All of these steps are performed without Alice communicating anything to Bob. This ensures that the correlations in $\rho_{AB}$, which are going to contribute towards our ability to distribute entanglement, are not unduly sullied (LOCC operations can increase correlations). Nevertheless, once this stage is complete, Alice and Bob are permitted to communicate in order to distil entanglement from $\rho_{ABC}$. In comparison to the original protocol of \cite{toby}, we have prevented the initial resource state from having some correlations with qubit $C$, meaning all the relevant information is contained within $\rho_{AB}$.

Under the restriction of $\rho_{AB}$ being Bell-diagonal, we describe the set of states $\Sigma_{\text{EDSS}}$ that cannot be used for EDSS. The ability to distribute entanglement for all other Bell-diagonal states is proven constructively, demonstrating the ubiquity of EDSS resources.

During this protocol, we only allow one qubit to be transmitted from Alice to Bob. Allowing further communication obviates the need for a resource ($\Sigma_{\text{EDSS}}$ is an empty set); with a two qubit protocol, one can always first distribute the optimal separable state for the one qubit protocol, and then perform the entanglement distribution with that, completely ignoring $\rho_{AB}$.

It must be emphasised that Protocol 1 makes this paper, in the main, incomparable to \cite{bound1,bound2} because \cite{bound1,bound2} effectively start with a state $\rho_{ABC}$ and ask what the correlations between $C|AB$ convey about entanglement distribution. However, at the start of the protocol, when Alice holds qubits $A$ and $C$, why should any such correlations constitute a relevant resource? Alice can change those correlations (within limits), optimising the protocol. This optimisation is incorporated in our discussion.


{\em Bell-diagonal and Graph States:} Throughout this paper, we restrict $\rho_{AB}$ to being Bell-diagonal, meaning that it can be written in the form
$$
\rho_{AB}=\smallfrac{4}\left(\identity+s_{01}XX+s_{10}ZZ+s_{11}YY\right)
$$
where $X$, $Y$ and $Z$ are the Pauli matrices.
The parameters $s_x$, $x\in\{0,1\}^2$ ($s_{00}=1$) or, alternatively, the ordered eigenvalues $\lambda_1\geq \lambda_2\geq\lambda_3\geq \lambda_4$ ($\sum_i\lambda_i=1$) encapsulate everything about the state. 
In studies of entanglement, the restriction to Bell-diagonal states is natural because all states can be made Bell-diagonal via $\Lambda_{\text{LOCC}}$. Having excluded classical communication, this is no longer true. Nevertheless, their structure will prove immensely useful. Such a state has
\begin{eqnarray*}
I_{\text{LOCC}}&=&1-H(\max(\half,\lambda_1)) \\
I_{\text{class}}&=&1+\sum_i\lambda_i\log_2\lambda_i+H(\lambda_1+\lambda_2)
\end{eqnarray*}
where $I_{\text{LOCC}}$ is the entanglement, $I_{\text{class}}$ is commonly referred to as the discord \cite{REQ,discord1,discord2}, and $H(x)$ is the binary entropy. Our initial state will always be separable, meaning $\lambda_1\leq\half$. Alice and Bob can perform some deterministic local operations in order to convert $\rho_{AB}$ into a canonical form with $|s_{01}|\geq |s_{10}|\geq |s_{11}|$. Also, if 
\begin{equation}
\prod_{x}s_x>0,	\label{eqn:signs}
\end{equation}
we can ensure that all $s_x$ are positive. Otherwise, all but one of the weights can be made positive.

\begin{theorem}
If $s_{11}=0$, EDSS is impossible, i.e.\ no distillable entanglement can be produced.
\label{th:1}
\end{theorem}
\begin{proof}
Let $\lambda(\rho)$ denote the spectrum of $\rho$, and $\rho^A$ be the partial transpose of $\rho$ on qubit $A$. In order for $C|AB$ to be separable, it must be that $\lambda(\rho_{ABC}^C)\geq 0$. In order for $A|BC$ to contain distillable entanglement, $\lambda(\rho_{ABC}^A)<0$ \cite{horo2}. Since there are no $Y$ terms present on $B$ (which cannot be changed by $U_{AC}$), $\lambda(\rho_{ABC}^A)=\lambda(\rho_{ABC}^{AB})=\lambda(\rho_{ABC}^C)$. Hence, imposing that $C|AB$ is separable ensures that $A|BC$ is non-distillable \footnote{This aspect generalises to arbitrary $\rho_{AB}$ -- the correlation matrix $T_{ij}=\Tr(\sigma_i\otimes\sigma_j\rho_{AB})$ for $i,j=1\ldots 3$ can be diagonalised by local operations. Assuming $T$ is not full rank, $T_{22}=0$, and Bob's marginals ${\vec n}\cdot{\vec \sigma}$ satisfy $n_2=0$, the same argument proves distillable EDSS is not possible.}.
\end{proof}

Having concluded that $s_{11}\neq 0$ is a necessary condition for EDSS, we want to investigate $s_{11}\neq 0$ (equivalently, $\lambda_1+\lambda_4\neq 0$), for which it will be sufficient to restrict $U_{AC}$ to being a controlled-phase gate. This enables the use of the graph state formalism, describing $\rho_{ABC}$ as (up to a Hadamard gate) being diagonal in the graph state basis of the linear 3 vertex graph.

In general, a graph $G$ is composed of a set of vertices $V$ and edges $E$. With each of the $N$ vertices, $i$, associate a qubit and define the stabilizer operators
$$
K_i=X_i\prod_{(i,j)\in E}Z_j,\qquad[K_i,K_j]=0.
$$
We will also use the notation $K_x$ for $x\in\{0,1\}^N$ to mean the product of all the $K_i$ for which the bits $x_i=1$ (and similarly for the Pauli operators).

The state $\ket{\psi^G}$ is the $+1$ eigenstate of each of the stabilizers, $K_i\ket{\psi^G}=\ket{\psi^G}$, and can be formally constructed by taking each qubit to be in the $\ket{+}=(\ket{0}+\ket{1})/\sqrt{2}$ state, and applying controlled-phase gates between every pair of qubits $(i,j)\in E$. The excited states
$$
\ket{\psi_x^G}=Z_x\ket{\psi}\qquad x\in\{0,1\}^N
$$
have eigenvalues $(-1)^{x_i}$ with each stabilizer $K_i$. Any graph-diagonal state can be written as
$$
\rho_G=\frac{1}{2^N}\sum_{x\in\{0,1\}^N}s_xK_x
$$
where $s_{00\ldots 0}=1$ and, in order to satisfy positivity,
$$
\sum_{x\in\{0,1\}^N}s_x(-1)^{x\cdot y}\geq 0\qquad\forall y\in\{0,1\}^N.
$$

As already indicated, the main graph that we are interested in is $G_3$, the chain of 3 vertices. We label the vertices of this chain as $1\equiv C$, $2\equiv A$, $3\equiv B$. For any vertex $R$, denote by $G_R$ the reduced subgraph of $G_3$ with vertex $R$ removed. So, $G_C$ is a two vertex chain, but $G_A$ is just two vertices, with no edges.

The partial transpose condition is particularly useful for detecting entanglement in bipartitions of $G$, and is implemented by manipulating the signs of the coefficients $s_x$, thereby leaving the eigenvectors unchanged \cite{Kay:2010,Kay:2010b}. Hence, with respect to a bipartition $z\in\{0,1\}^N$ (bit $z_i$ specifies the partition of vertex $i$), the state is non-positive under the partial transpose (NPT) if
$$
\min_{y\in\{0,1\}^N}\!\sum_{x\in\{0,1\}^N}\!\!\!s_x(-1)^{x\cdot y}(-1)^{\sum_{(i,j)\in E}x_ix_j(z_i\oplus z_j)}\leq0,
$$
otherwise it is positive under the partial transpose (PPT) with respect to the bipartition $z$. The following two lemmas prove that, for the class of states we're interested in, the partial transpose condition conveys everything we need to know; they add sufficiency to the necessity conditions that always exist between, first, NPT and distillability and, secondly, PPT and separability.

\begin{lemma} If a state $\rho_{G_3}$ has a bipartition which is NPT, the state is distillable with respect to that bipartition.
\label{lem:1}
\end{lemma}
\begin{proof}
For a non-trivial bipartition of $G_3$, there must be one qubit ($R$) on one side of the partition, and two qubits on the other. Taking the partial transpose over $R$ and assuming the eigenvector with negative eigenvalue, $\lambda_R$, to be $\ket{\psi^{G_3}_x}$, we have that
$$
\ket{\psi^{G_3}_x}=Z_x\frac{1}{\sqrt{2}}\left(\ket{0}_R\ket{\psi^{G_R}}+\ket{1}Z_E\ket{\psi^{G_R}}\right)
$$
where $Z_E$ is a product of $Z$s on the neighbours of $R$. Given that the partial transpose operation didn't act on the two-qubit bipartition, this means that the states that give the negative eigenvalue exist in the subspace $\{\ket{\psi_{x\backslash x_R}^{G_R}}, Z_E\ket{\psi_{x\backslash x_R}^{G_R}}\}$. By applying the projection
$$
P=\left(\ket{0}\bra{\psi_{x\backslash x_R}^{G_R}}+\ket{1}\bra{\psi_{x\backslash x_R}^{G_R}}Z_E\right),
$$
we produce a two-qubit state $\rho_{\text{SUC}}$ with a negative eigenvalue $\lambda_R/p$ under partial transposition, where $p$ is the success probability of the projection. Once localised to a pair of qubits, NPT is sufficient for the distillation of entanglement \cite{horo}.
\end{proof}
If the state at the end of our protocol is NPT with respect to $A|BC$, Alice and Bob can use multiple copies to extract that entanglement by first projecting into a two qubit state (which is clearly optimal), and then distilling the successfully projected copies \footnote{In general, project onto the Schmidt basis of the negative eigenvector of the partial transpose \cite{bound2}.}.

\begin{lemma}
If a state $\rho_{G_3}$ has a bipartition which is PPT, the state is separable with respect to that bipartition.
\label{lem:2}
\end{lemma}
\begin{proof}
Utilising the techniques in \cite{Kay:2010,Kay:2010b}, the state can be expanded in terms of the stabilizers, which are grouped into sets that have simultaneous eigenvectors which are separable with respect to the bipartition. This consists of grouping the terms according to their Pauli operator on $R$. For $R=C$, $8\rho_{ABC}$ can be expressed as
\begin{eqnarray*}
&8\lambda_{C|AB}\identity +&\\
&s_{11}Z_AX_B+Z_C(s_{10}X_AZ_B+s_{11}Y_AY_B)+s_{10}\identity +&\\
&ss_{01}Z_AX_B+sX_C(Z_A+s_{01}X_B)+s\identity +&\\
&ss_{11}Z_AX_B+sY_C(s_{10}Y_AZ_B-s_{11}X_AY_B)+ss_{10}\identity +&\\
&(s_{01}\!-\!s_{11}\!-\!ss_{01}\!-\!ss_{11})Z_AX_B\!+\!|s_{01}\!-\!s_{11}\!-\!ss_{01}\!-\!ss_{11}|\identity&
\end{eqnarray*}
where $\lambda_{C|AB}$ is the minimum eigenvalue of $\rho_{ABC}^C$. Every line represents a separable state, provided that bipartition is PPT. Similarly, for $R=A$,
\begin{eqnarray*}
&8\lambda_{A|BC}\identity +&\\
&ss_{01}X_BX_C+Z_A(sX_C\!+\!s_{01}X_B)+\identity(s\!+\!s_{01}\!-\!ss_{01}) +&\\
&X_A(s_{10}Z_BZ_C-ss_{11}Y_BY_C)+\identity(s_{10}+ss_{11}) +&\\
&Y_A(ss_{10}Z_BY_C+s_{11}Y_BZ_C)+\identity(ss_{10}+s_{11}).&
\end{eqnarray*}
PPT guarantees separability.
\end{proof}

\begin{theorem}
For a Bell-diagonal state $\rho_{AB}$ with $(\lambda_1+\lambda_4)>\half$ (i.e. all $s_x>0$), distillable entanglement can always be produced via EDSS.
\label{th:2}
\end{theorem}
\begin{proof}
Given that, for $U_{AC}$ being a controlled-phase gate,
\begin{eqnarray}
8\lambda_{C|AB}\!\!&=&\!\!1\!-\!s_{10}\!-\!s(1\!+\!s_{10})-|s_{01}\!-\!s_{11}\!-\!ss_{01}\!-\!ss_{11}| \label{eqn:C}\\
8\lambda_{A|BC}\!\!&=&\!\!1\!-\!s_{01}\!-\!s_{10}\!-\!s_{11}\!-\!s(1\!-\!s_{01}\!+\!s_{10}\!+\!s_{11}), \label{eqn:A}
\end{eqnarray}
if Alice prepares $\rho_C$ with $s=\min(\lambda_4/\lambda_3,\lambda_2/\lambda_1)$, then $\lambda_{C|AB}=0$. Lemma \ref{lem:2} ensures that $C$ remains separable from $AB$. 
Since $\lambda_1+\lambda_4>\half$, one can show that
$$
\lambda_{A|BC}=\frac{1}{4}\left(1-2\lambda_1-s(1-2\lambda_2)\right)<0,
$$
meaning that the entanglement distribution protocol is successful via Lemma \ref{lem:1}, which specifies that Bob localises the entanglement by performing a projection
\begin{equation}
P=\ket{0}\bra{++}+\ket{1}\bra{--}	\label{eqn:proj}
\end{equation}
on his two qubits. With probability
$$
p=\half\left(1+s(2\lambda_1+2\lambda_2-1)\right)>\half.
$$
a two-qubit state with negative eigenvalue $\lambda_{A|BC}/p$ under partial transpose is produced.
\end{proof}
\begin{lemma}
For $\rho_{AB}$ with $\lambda_1+\lambda_4<\half$, distillable entanglement can always be produced via EDSS.
\end{lemma}
\begin{proof}
In the regime $\lambda_1+\lambda_4<\half$, one of the coefficients $s_x$ is negative. If $|s_{01}-s_{11}-ss_{01}-ss_{11}|>0$, make $s_{11}$ the negative coefficient, otherwise make it $s_{01}$. By replacing the negative coefficient $s_x$ with $-|s_x|$, Eqns.\ (\ref{eqn:C}) and (\ref{eqn:A}) exchange roles. Hence the contents of Theorem \ref{th:2} are reproduced, i.e.\ the same value of $s$ makes $A|BC$ separable, leaving $C|AB$ entangled, so Alice just chooses to send qubit $A$ instead of qubit $C$ (although the basis of Bob's projection must change as a result).
\end{proof}
We conclude that $s_{11}\neq 0$ is necessary and sufficient for a Bell diagonal state to be useful for EDSS. The optimal state for the outlined protocol has $s_{01}=\half$ and $s_{10}=s_{11}=\smallfrac{4}$, yielding $\lambda_{A|BC}=-\smallfrac{16}$. Although the example in \cite{toby} had $\lambda_{A|BC}=-\smallfrac{6}$, that protocol had the advantage of correlations with qubit $C$ in the initial resource.



{\em Quantifying Resources:} Our results exactly characterise the set of Bell-diagonal states for which entanglement distribution is impossible, $\Sigma_{\text{EDSS}}$. We have also given an explicit protocol which serves to lower bound the amount of entanglement, $I_{\text{LOCC}}$, that can be generated for a given resource $\rho_{AB}$, 
\begin{equation}
I_{\text{LOCC}}\geq p(1-H((1+s)\lambda_1/(2p))). \label{eqn:lower}
\end{equation}
How good a bound is this? The na\"ive expectation that $I_{\text{EDSS}}=1-H(\lambda_1+\lambda_4)$ should upper bound the ability to transfer entanglement does not hold because it cannot be guaranteed that the closest state $\sigma_{AB}\in\Sigma_{\text{EDSS}}$ becomes separable across $C|AB$ at the same value of $s$ that $\rho_{AB}$ does. Nevertheless, the proximity of the two states predicts a good approximation, and numerically it is closely related to the lower bound.

To find a rigorous bound, we start with $\rho_{AB}$ and attempt to dephase qubit $B$ to form a $\sigma_{AB}\in\Sigma_{\text{EDSS}}$. This can only be achieved with $\sigma_{AB}=\half(\rho_{AB}+Z_B\rho_{AB}Z_B)$. The dephasing only removes entanglement, so PPT of $\rho_{C|AB}$ is preserved. Hence, $\sigma_{ABC}$ formed from $\sigma_{AB}$ is PPT about the $A|BC$ partition as well. The eigenvalues $\chi_i$ of $\sigma_{AB}$ satisfy $\chi_1=\chi_2$ (the property of Bell-diagonal states in $\Sigma_{\text{class}}$), meaning that $I_{\text{LOCC}}\leq I_{\text{class}}$. Although we relied on the PPT criterion, one can readily check that the states $\sigma_{ABC}$ in this class are separable. This bound is rather weak, as we can see from examining the subset of states drawn from $\Sigma_{EDSS}$. We know that these must all have $I_{\text{LOCC}}=0$. While for some of these states $I_{\text{class}}=0$, $\Sigma_{EDSS}$ also includes the separable states with the largest discord ($\lambda_1=\half$, $\lambda_2=\lambda_3=\smallfrac{4}$). The difference arises because the states $\rho_{AB}$ with 0 discord have two non-zero values of $s_x$, whereas we have proven that any $\rho_{AB}$ with a single $s_x$ equal to zero cannot be used for EDSS. However, it remains the only bound that we have succeeded in proving for arbitrary $U_{AC}$.

There is strong evidence that our lower bound is optimal, beyond its numerical closeness to $I_{\text{EDSS}}$. Extending our argument from Theorem \ref{th:1}, for every $U_{AC}$,
\begin{eqnarray}
\lambda(\rho^C_{ABC})&=&\lambda(\rho^A_{ABC}-\half s_{11}U_{AC}Y_AY_B\otimes \rho_CU_{AC}^\dagger)	\nonumber\\
\left|\lambda_{A|BC}-\lambda_{C|AB}\right|&\leq&\smallfrac{4}(1+s)|s_{11}|.	\label{eqn:diff}
\end{eqnarray}
In our explicit protocol, where $s=\min(\lambda_4/\lambda_3,\lambda_2/\lambda_1)$, then if $s=\lambda_4/\lambda_3$, this bound is saturated, meaning that $\lambda_{A|BC}$ is as small as possible given $\lambda_{C|AB}\geq 0$, and we already know that the entanglement localisation and distillation steps are optimal. The reason that this isn't a full optimality proof is that we have been unable to show that some other $U_{AC}$ cannot be less entangling on $\rho_{C|AB}$, allowing a larger $s$, while simultaneously being more entangling on $\rho_{A|BC}$. Nonetheless, the scope for improvement, $\smallfrac{4}(1-\lambda_4/\lambda_3)s_{11}$, is extremely limited.

{\em Noise Tolerance:} One of the useful features of the protocol is that, apart from specifying $s$, Alice and Bob act independently of the choice of the state. This means that even if the state is changed by noise, the protocol can still function. Does this impart the protocol with increased noise tolerance over the direct distribution of an entangled state? Imagine, for instance, that the qubits that Alice and Bob hold are well protected from noise, but the channel through which qubit $C$ is sent is noisy. Let us compare EDSS using a state $\rho_{AB}$ with $\lambda_1=\half$, and Alice directly transmitting one half of a maximally entangled pair to Bob. For phase flip errors, both protocols fail at the same point, when the probability of a flip is $\half$. Alternatively, depolarising noise
$$
\mathcal{E}(\rho)=(1-\frac{3q}{4})\rho+\frac{q}{4}\left(X\rho X+Y\rho Y+Z\rho Z\right)
$$
can be tolerated (i.e.\ non-0 distillable entanglement is distributed in a heralded way) provided
$
q<2s/(2s+1)
$
for EDSS and $q<\smallfrac[2]{3}$ for the maximally entangled state. Since $s<1$, distribution with the maximally entangled state is always more successful. EDSS is not more noise tolerant than direct transmission of entanglement. Of course, a maximally entangled state has a lot of entanglement that it can afford to lose. A fairer comparison might be with the transmission of one qubit from $\rho_{\text{SUC}}$. This is significantly less robust than using EDSS.

{\em Conclusions:} The set of Bell-diagonal states which cannot be used for the protocol of entanglement distribution via separable states have been exactly classified as those with $\lambda_1+\lambda_4=\half$. Any state not satisfying this can distribute some entanglement, making it an extremely common feature of quantum correlations, in contrast to the limited number of examples in the literature \cite{toby,bound2}. These correlations constitute a useful resource, and the amount of entanglement that can be distributed has been lower bounded, Eqn.\ (\ref{eqn:lower}). We conjecture this to be optimal, and have provided supporting evidence. Beyond resolving this conjecture, in the future it will be interesting to extend these proofs to all bipartite mixed states.

We have also seen that the quantum discord readily arises as a weak upper bound to the amount of entanglement that can be distributed. Much tighter bounds would be useful. To this end, while not a literal bound, $I_{\text{EDSS}}$ may prove useful in approximating the potential of a given system. While \cite{bound1,bound2} found a bound based on the discord, they considered a significantly different protocol, and it is a different discord (based on a bipartitioning of the tripartite system) that they used -- one which we have argued is uninformative from a resource perspective. The exception to this is Theorem 3 in \cite{bound2}, which gives the same upper bound as presented here.

\end{document}